\documentclass[runningheads]{llncs}
\usepackage[T1]{fontenc}
\usepackage{graphicx}
\usepackage{color}
\usepackage[ruled,vlined]{algorithm2e}
\usepackage{amsmath,amssymb,amsfonts}
\usepackage{algpseudocode}
\usepackage{booktabs}
\usepackage{graphicx}
\usepackage{subcaption}
\usepackage{multirow}
\usepackage{tabularx}
\usepackage{svg}
\usepackage{hyperref}
\usepackage{esvect}
\usepackage{xcolor}
\usepackage{listings}
\usepackage{multirow}
\usepackage{float}
\usepackage{enumitem}
\usepackage{enumitem} 
\usepackage{listings}
\usepackage{xcolor}
\usepackage{hyperref}

\newcommand{\ALG}{\mathrm{ALG}}
\newcommand{\OPT}{\mathrm{OPT}}
\newcommand{\CR}{\mathrm{CR}}

\begin{document}
 
\title{Competitive Analysis of Online Facility Assignment Algorithms on Discrete Grid Graphs: Performance Bounds and Remediation Strategies}
 
\titlerunning{Expected Cost Analysis of OFA on regular polygons}
\author{
Lamya Alif \and
Raian Tasnim Saoda \and
Sumaiya Afrin \and
Md. Rawha Siddiqi Riad \and
Md. Tanzeem Rahat \and
Md Manzurul Hasan
}

\institute{
Department of Computer Science and Engineering\\
American International University-Bangladesh\\
Dhaka, Bangladesh\\
22-46350-1@student.aiub.edu,
22-46331-1@student.aiub.edu,
22-46304-1@student.aiub.edu,
22-46678-1@student.aiub.edu
tanzeem.rahat@aiub.edu,
manzurul@aiub.edu
}

\authorrunning{Alif, L., Saoda, R.T., Afrin, S., Riad M.R.S., Rahat, M.T. \and Hasan, M.M.}
\maketitle
\begin{abstract}
We study the \emph{Online Facility Assignment} (OFA) problem on a discrete $r\times c$ grid graph under the standard model of Ahmed, Rahman, and Kobourov: a fixed set of facilities is given, each with limited capacity, and an online sequence of unit-demand requests must be irrevocably assigned upon arrival to an available facility, incurring Manhattan ($L_1$) distance cost. We investigate how the discrete geometry of grids interacts with capacity depletion by analyzing two natural baselines and one capacity-aware heuristic. 
First, we give explicit adversarial sequences on grid instances showing that purely local rules can be forced into large competitive ratios: (i) a capacity-sensitive weighted-Voronoi heuristic (\textsc{CS-Voronoi}) can suffer cascading \emph{region-collapse} effects when nearby capacity is exhausted; and (ii) nearest-available \textsc{Greedy} (with randomized tie-breaking) can be driven into repeated long reassignments via an \emph{oscillation} construction. These results formalize geometric failure modes that are specific to discrete $L_1$ metrics with hard capacities.
Motivated by these lower bounds, we then discuss a semi-online extension in which the algorithm may delay assignment for up to $\tau$ time steps and solve each batch optimally via a min-cost flow computation. We present this batching framework as a remediation strategy and delineate the parameters that govern its performance, while leaving sharp competitive guarantees for this semi-online variant as an open direction.
\end{abstract}
\section{Introduction}

The \emph{Online Facility Assignment} (OFA) problem models a basic tension in resource-constrained systems: decisions must be made irrevocably as requests arrive, while the resources that serve them are limited. In the OFA model of Ahmed, Rahman, and Kobourov~\cite{ahmed2020}, a set of facilities is fixed in advance, each facility has a finite capacity, and an online sequence of unit-demand requests must be assigned immediately to an available facility. The cost of assigning request $u$ to facility $f$ is the metric distance $d(u,f)$. This abstraction captures practical settings such as dispatching orders to warehouses or restaurants, routing jobs to processing units, or assigning tasks to nearby service stations, where serving locally is desirable but consuming capacity too aggressively can harm future decisions.

In this work we focus on \emph{discrete grid graphs}. We model the underlying metric space as an $r\times c$ grid and measure distance by the Manhattan ($L_1$) metric, which is appropriate when movement is restricted to axis-aligned steps (e.g., city blocks or on-chip routing). The discrete nature of the grid is not a benign modeling choice: it creates large tie regions and non-smooth changes in nearest-available structure, which interact sharply with hard capacity constraints. Small changes in remaining capacity can therefore lead to discontinuous shifts in which facility is chosen by local rules, enabling adversarial sequences that repeatedly force expensive reassignments.

The standard way to evaluate an online algorithm is via the \emph{competitive ratio}. For an input sequence $\sigma$, let $\ALG(\sigma)$ denote the total cost incurred by the online algorithm and let $\OPT(\sigma)$ denote the minimum possible cost of an offline algorithm that knows $\sigma$ in advance and respects the same capacity constraints. The competitive ratio is $\sup_{\sigma}\ALG(\sigma)/\OPT(\sigma)$ for deterministic algorithms, and the worst-case expected ratio for randomized algorithms against an oblivious adversary. Competitive analysis is particularly relevant for OFA because even natural strategies can behave well on typical inputs yet fail dramatically on carefully constructed sequences~\cite{ahmed2020}.

Our goal is to understand, on grid metrics, which aspects of local decision-making are inherently fragile under capacities, and which kinds of algorithmic modifications are necessary to avoid such geometric failure modes.

\subsection*{Related Work}

Online facility allocation problems have been studied both in the classical \emph{online facility location} setting and in the closely related \emph{online facility assignment} setting. The OFA model of Ahmed, Rahman, and Kobourov~\cite{ahmed2020} formalizes assignment to a fixed set of capacitated facilities under irrevocable online decisions and metric costs, and provides baseline guarantees together with adversarial examples that expose how capacity depletion can amplify the cost of local choices. For grid metrics in particular, Muttakee, Ahmed, and Rahman~\cite{muttakee2020} develop worst-case bounds and constructions illustrating that structured discrete geometries can degrade greedy and deterministic assignment rules. 

A separate line of work studies geometric partitions and Voronoi-type structures as algorithmic primitives. Metric-sensitive and weighted Voronoi formulations appear broadly in computational geometry and visualization, including metric-driven variants~\cite{nair2025} and Voronoi treemaps~\cite{nocaj2012}. While these works are not OFA analyses, they motivate capacity-aware geometric heuristics that adjust partitions based on non-uniform weights. 

Finally, delay and batching have been explored in online matching and assignment models in which decisions may be postponed to improve global structure. Min-cost matching with delays was studied by Ashlagi et al.~\cite{ashlagi2017}, and randomized approaches for metric matching were developed by Meyerson et al.~\cite{meyerson2006}. In the data-management literature, Tong et al.~\cite{tong2016} empirically study online spatial matching under real-time constraints. Although these models differ from OFA, they suggest that bounded postponement or batching can mitigate adversarial effects by allowing limited lookahead and global optimization within a window.

Working under the OFA model of Ahmed et al.~\cite{ahmed2020} on $r\times c$ grid graphs with Manhattan distance, we formalize discrete geometric failure modes for natural local assignment rules: we present explicit adversarial sequences showing that (i) a capacity-aware weighted-Voronoi heuristic can undergo \emph{region-collapse} and repeated \emph{boundary oscillation} once nearby facilities deplete, and (ii) nearest-available greedy assignment (including randomized tie-breaking) can be forced into repeated long reassignments through an \emph{oscillation} construction on grids. Motivated by these lower bounds, we also outline a semi-online batching-and-min-cost-flow framework as a remediation strategy, clearly separating it as an extension beyond the standard OFA model and identifying the parameters that would govern its performance.

\subsection*{Key Contributions}

\renewcommand{\labelenumi}{\Roman{enumi}.}
\begin{enumerate}
    \item
    We give explicit adversarial constructions on $r\times c$ grids showing how hard capacities can trigger abrupt \emph{region-collapse} and \emph{boundary oscillation} phenomena for capacity-aware geometric heuristics.

    \item
    We provide an oscillation-style adversarial sequence demonstrating that nearest-available greedy assignment (including randomized tie-breaking) can be driven into repeatedly paying large grid distances compared to $\OPT$.

    \item 
    We propose a batching-based remediation approach that solves each batch via a min-cost flow computation under a bounded waiting window, and we delineate how such postponement can counteract the above failure modes; establishing tight competitive guarantees for this semi-online extension remains open.
\end{enumerate}

\section{Preliminaries}

\paragraph{Metric space (grid).}
We work on an $r\times c$ grid graph $G=(V,E)$ whose vertices are grid points and whose edges connect axis-adjacent vertices. Distances are measured by the shortest-path distance on $G$, which equals the Manhattan ($L_1$) distance when vertices are identified with integer coordinates:
\[
d\big((x,y),(x',y')\big)=|x-x'|+|y-y'|.
\]

\paragraph{Facilities and capacities (OFA model).}
A set of facilities $F=\{f_1,\dots,f_k\}\subseteq V$ is fixed in advance. Each facility $f$ has an integer capacity $c(f)\ge 1$, meaning it can serve at most $c(f)$ requests in total. We consider the standard unit-demand setting: each arriving request consumes one unit of capacity at its assigned facility. Let $\mathrm{remcap}_t(f)$ denote the remaining capacity of facility $f$ immediately before processing the $t$-th request.

\paragraph{Requests, assignments, and cost.}
An input sequence $\sigma=(u_1,u_2,\dots,u_n)$ of requests arrives online, where each $u_t\in V$. Upon arrival of $u_t$, an online algorithm must irrevocably choose a facility $f_t\in F$ with $\mathrm{remcap}_t(f_t)>0$, assign $u_t$ to $f_t$, and then decrement $\mathrm{remcap}_t(f_t)$ by one. The incurred cost for this assignment is the metric distance $d(u_t,f_t)$. The total cost of the algorithm on $\sigma$ is
\[
\ALG(\sigma)=\sum_{t=1}^n d(u_t,f_t).
\]
An \emph{offline optimal} solution $\OPT(\sigma)$ knows the entire sequence $\sigma$ in advance and minimizes total cost subject to the same capacity constraints~\cite{ahmed2020}.

\paragraph{Competitive ratio and adversaries.}
For a deterministic algorithm, the competitive ratio is
\[
\CR(\ALG)=\sup_{\sigma}\frac{\ALG(\sigma)}{\OPT(\sigma)}.
\]
For a randomized algorithm, we consider an \emph{oblivious adversary} that fixes $\sigma$ in advance, and define
\[
\CR(\ALG)=\sup_{\sigma}\frac{\mathbb{E}[\ALG(\sigma)]}{\OPT(\sigma)},
\]
where the expectation is over the algorithm's internal randomness. This is the standard benchmark used for randomized online algorithms in OFA-style settings.

\paragraph{A capacity-aware geometric heuristic (CS-Voronoi).}
In addition to standard baselines such as nearest-available greedy assignment, we study a simple \emph{heuristic} that biases choices toward facilities with larger remaining capacity. Given a parameter $\alpha>0$, define the score
\[
D_t(u,f)=d(u,f)-\alpha\cdot \mathrm{remcap}_t(f),
\]
and upon arrival of $u_t$, assign it to an available facility minimizing $D_t(u_t,f)$. Intuitively, the subtractive term favors facilities with more remaining capacity, encouraging load balancing. Importantly, in a discrete $L_1$ grid, tie regions and integer distance gaps can cause the argmin of $D_t(\cdot,\cdot)$ to change abruptly as $\mathrm{remcap}_t(\cdot)$ decreases by one. In later sections we use explicit adversarial sequences to illustrate how such discontinuities can trigger cascading reassignment behavior under hard capacities.

\section{CS-Voronoi Algorithm}

Voronoi-type partitions are a natural geometric primitive for assignment: in a static setting, each point is served by its nearest facility. Under hard capacities, however, \emph{proximity alone} can be misleading. A facility that is geometrically closest may be close to depletion, and committing early requests to it can force later requests to travel much farther once the local capacity is exhausted. Motivated by this tension, we study a simple \emph{capacity-aware} geometric heuristic that adjusts the effective distance to a facility based on its remaining capacity. While weighted Voronoi ideas are broadly useful in geometric modeling~\cite{nair2025,nocaj2012}, our use here is purely algorithmic: we treat the resulting rule as a deterministic heuristic within the standard OFA model~\cite{ahmed2020} and analyze its failure modes on discrete grids.

\subsection*{Heuristic Definition (Capacity-Sensitive Weighted Distance)}

Fix a parameter $\alpha>0$. At time $t$, let $\mathrm{remcap}_t(f)$ denote the remaining capacity of facility $f$ before processing request $u_t$. The \textsc{CS-Voronoi} heuristic assigns $u_t$ to an available facility minimizing the \emph{capacity-sensitive score}
\begin{equation}
\label{eq:cs-voronoi-score}
D_t(u_t,f) \;=\; d(u_t,f) \;-\; \alpha\cdot \mathrm{remcap}_t(f),
\end{equation}
breaking ties deterministically (e.g., by facility index). Intuitively, the subtractive term biases assignments toward facilities with larger remaining capacity, encouraging load spreading when geometric distances are comparable.

\paragraph{Dynamic partition viewpoint.}
For any fixed time $t$, the rule~\eqref{eq:cs-voronoi-score} induces a partition of the grid into regions served by each facility, analogous to an additively weighted Voronoi diagram. Unlike classical Voronoi partitions, these regions change over time because $\mathrm{remcap}_t(\cdot)$ decreases with assignments. On a discrete $L_1$ grid, distance differences are integer-valued, so changing $\mathrm{remcap}_t(f)$ by one unit can shift the argmin of~\eqref{eq:cs-voronoi-score} abruptly, especially near large tie regions. This discrete sensitivity is the root cause of the failure modes described next.

\subsection*{Geometric Failure Modes on Grids}

We highlight two recurring adversarial patterns that exploit capacity depletion together with discrete $L_1$ geometry. These patterns are not presented as universal theorems about the heuristic; rather, they are \emph{constructive failure modes} that later sections instantiate via explicit request sequences under the OFA model~\cite{ahmed2020}.

\paragraph{Zone-Collapse (Region Collapse after Depletion).}
When a facility that serves a large surrounding region becomes fully depleted, it disappears from the set of available choices. Requests that would have been served locally inside its region must then be reassigned to the next-best available facility, which can be substantially farther on the grid. This produces an abrupt jump in cost for a contiguous set of subsequent arrivals (Figure~\ref{fig:1}).

\begin{figure}[h]
\centering
\includegraphics[width=0.6\textwidth]{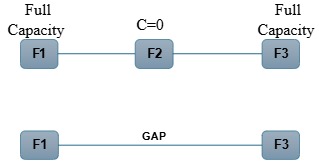}
\caption{Zone-Collapse: once a central facility depletes, a large nearby region loses its local option and subsequent requests incur a sharp distance increase.}
\label{fig:1}
\end{figure}

\paragraph{Boundary-Oscillation (Discrete Bisector Flips).}
Between two competing facilities, the set of points where the scores are nearly tied can form a thick, staircase-like band under $L_1$. Because the score includes the integer-valued term $\mathrm{remcap}_t(\cdot)$, decrementing capacity by one unit can flip the minimizer for many points in this band. An adversary can place consecutive requests near such a band so that the heuristic alternates between facilities in a way that repeatedly forces avoidable long assignments (Figure~\ref{fig:2}).

\begin{figure}[H]
\centering
\includegraphics[width=0.65\textwidth]{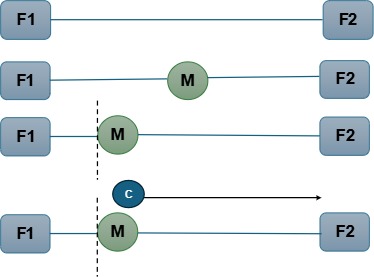}
\caption{Boundary-Oscillation: small capacity changes can flip the minimizing facility across a wide staircase bisector, enabling repeated costly redirections.}
\label{fig:2}
\end{figure}

\subsection*{When the Heuristic Behaves Well}

Although \textsc{CS-Voronoi} has adversarial weaknesses, it can coincide with the offline optimum on benign inputs. The following statement is a sanity check (not a performance guarantee) that clarifies when the heuristic reduces to nearest-facility assignment.

\begin{proposition}[Best-Case Behavior]
\label{prop:cs-best-case}
If throughout the execution the remaining capacities of all available facilities stay equal (or differ by a constant that does not affect tie-breaking), then minimizing $D_t(u,f)$ is equivalent to minimizing the physical distance $d(u,f)$. In particular, on sequences where the nearest available facility is always feasible for $\OPT$, the heuristic makes the same per-request choices as nearest-available assignment.
\end{proposition}

\begin{proof}
If $\mathrm{remcap}_t(f)$ is the same for all available $f$ (or differs by a constant independent of $f$), then the term $-\alpha\cdot \mathrm{remcap}_t(f)$ is the same additive offset for every available facility. Hence the minimizer of $D_t(u,f)$ is exactly the minimizer of $d(u,f)$, with identical tie structure.
\end{proof}

\subsection*{What We Prove About CS-Voronoi}

Consistent with our focus on \emph{failure modes on grids}, we do \emph{not} claim a tight worst-case competitive ratio for \textsc{CS-Voronoi} in general. Instead, we use explicit adversarial constructions to demonstrate that the heuristic can incur a large competitive ratio on grid instances due to Zone-Collapse and Boundary-Oscillation. These lower bounds complement known worst-case phenomena for OFA on structured metrics~\cite{muttakee2020,ahmed2020} and motivate remediation strategies discussed later. Figure~\ref{fig:4} and Table~1 are retained as illustrative summaries of the mechanism and the comparative picture.

\begin{figure}[h]
\centering
\includegraphics[width=0.75\textwidth]{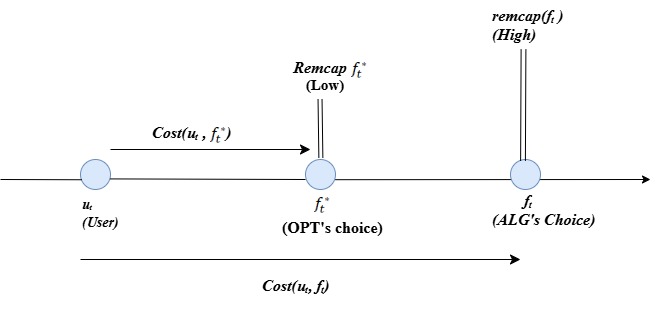}
\caption{Illustration of capacity-sensitive scoring and how discrete capacity changes can alter the minimizing facility.}
\label{fig:4}
\end{figure}

\section{Greedy Baselines on Grids}
\label{sec:greedy}

This section presents the most direct OFA baselines on grid metrics and explains a canonical adversarial failure pattern that arises from \emph{hard capacities} together with \emph{large tie regions} under the Manhattan distance. Throughout, we use the OFA model from Section~2: facilities are fixed and capacitated, each request has unit demand, assignments are irrevocable upon arrival, and the assignment cost is the $L_1$ grid distance~\cite{ahmed2020}.

\subsection{Nearest-Available Greedy}
\label{subsec:na-greedy}

\paragraph{Algorithm.}
Upon arrival of request $u_t$, the \emph{Nearest-Available Greedy} algorithm assigns $u_t$ to an available facility of minimum distance:
\[
f_t \in \arg\min_{f\in F:\,\mathrm{remcap}_t(f)>0} d(u_t,f),
\]
breaking ties deterministically (e.g., by facility index). This is the natural baseline that always serves each request as locally as possible given current availability.

\subsection{Randomized Greedy (Random Tie-Breaking)}
\label{subsec:rg}

Because the $L_1$ metric on grids creates many ties (multiple facilities can be at the same minimum distance from a request), a standard randomized variant replaces deterministic tie-breaking with a uniform random choice among the nearest available facilities:
\[
f_t \sim \mathrm{Unif}\Big(\arg\min_{f\in F:\,\mathrm{remcap}_t(f)>0} d(u_t,f)\Big).
\]
We evaluate randomized algorithms against an \emph{oblivious adversary} (the request sequence is fixed in advance) as in Section~2.

\subsection{A Canonical Failure Pattern: The Oscillation Trap}
\label{subsec:oscillation}

We now describe a simple geometric-capacity mechanism that can force Randomized Greedy to incur a large cost relative to $\OPT$ on grid instances. The key idea is that a single random tie-break can \emph{consume} a strategically important unit of local capacity; if future demand arrives at that location, the algorithm is then forced to ``jump'' to a far facility. We refer to this mechanism as an \emph{oscillation trap} because the adversary can repeatedly steer assignments across the grid by alternating which local capacity is removed.

\paragraph{Mechanism.}
Consider two facilities $F(L)$ and $F(R)$ with unit capacities. Let the first request $R_1$ arrive at a position that is equidistant from both facilities. Randomized Greedy flips a coin and assigns $R_1$ to one of them, say $F(L)$, thereby exhausting $F(L)$. The adversary then issues a second request $R_2$ at (or very near) the location of the now-full facility $F(L)$. Since $F(L)$ is unavailable, the algorithm must assign $R_2$ to $F(R)$, paying the full cross-distance between the two facilities. Figure~\ref{fig:rg-oscillation} illustrates this ``fill-then-hit'' pattern.

\begin{figure}[t]
\centering
\includegraphics[width=0.95\textwidth]{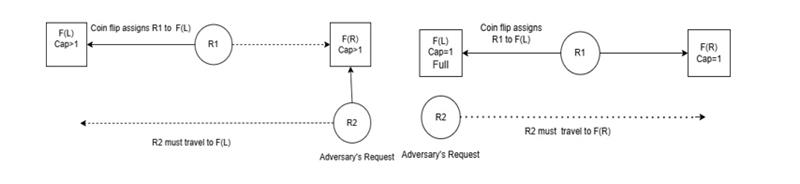}
\caption{\textbf{Oscillation trap for Randomized Greedy under hard capacities.}
A first request $R_1$ arrives at a location equidistant from two unit-capacity facilities $F(L)$ and $F(R)$.
Randomized Greedy breaks the tie by a coin flip and fills one facility.
An adversary then releases a second request $R_2$ at the location of the filled facility, forcing assignment to the other facility and incurring the full cross-distance.}
\label{fig:rg-oscillation}
\end{figure}

\paragraph{Distance scale on an $r\times c$ grid.}
On the $r\times c$ grid, the largest forced ``jump'' is governed by the grid diameter under $L_1$:
\[
D \;=\; (r-1)+(c-1).
\]
Thus, any construction that repeatedly forces assignments between far-apart facilities can make the online cost scale with $D$ per forced jump, while $\OPT$ may be able to keep most assignments local by reserving scarce nearby capacity for future requests~\cite{muttakee2020,ahmed2020}.

\paragraph{What we use this for.}
In later sections, we instantiate this pattern as an explicit adversarial request sequence on grid instances (with a small number of well-separated facilities) to obtain clean lower bounds for nearest-available baselines. The role of this section is to fix the baseline algorithms precisely and to isolate the underlying geometric-capacity mechanism that makes them vulnerable on discrete grids.

\section{Batching + Min-Cost Flow (BMCF) Algorithm}
\label{sec:bmcf}

This section presents a widely-used semi-online baseline for the Online Facility Assignment (OFA) model of~\cite{ahmed2020}: \emph{batch requests for a short window}, then compute the \emph{optimal} capacitated assignment for that batch using a min-cost flow (MCF) formulation. BMCF is attractive because it is conceptually simple, implementable, and often improves over per-request greedy rules on bursty demand. However, the key message of this paper is that \emph{even optimal within-batch decisions can fail badly on grid metrics under hard capacities}, because the algorithm does not coordinate across batches.

\subsection{Model reminder (OFA on a grid)}
We work on an $r\times c$ grid graph $G=(V,E)$ with shortest-path distance $d(\cdot,\cdot)$ under the $L_1$ (Manhattan) metric. Let $F=\{f_1,\dots,f_k\}\subseteq V$ be the set of facilities, each with an integer capacity $C(f)\ge 0$ (often identical, but not required). Requests arrive over discrete time steps $t=1,2,\dots$. A request is a vertex $u_t\in V$ that must be assigned irrevocably to some facility $f$ with remaining capacity at the time of assignment, paying cost $d(u_t,f)$. Each assignment decreases the facility's remaining capacity by 1.

The offline optimal solution $\OPT$ sees the entire request sequence in advance and minimizes total cost subject to the same capacity constraints. The competitive ratio compares the online/semi-online algorithm to $\OPT$ on worst-case sequences (as standard in online optimization).

\subsection{Why batching helps (the intended benefit)}
Purely online rules make irrevocable decisions with essentially zero context. Batching introduces a small amount of \emph{lookahead} by waiting briefly and then solving a \emph{global optimization problem} on the set of requests observed in that window. Informally:
\begin{itemize}[leftmargin=*]
    \item If arrivals are \emph{bursty} and locally correlated, the batch optimizer can coordinate assignments to avoid unnecessary conflicts inside the burst.
    \item If multiple facilities are nearby, the optimizer can split load across them \emph{within the batch} rather than committing too early to one facility.
\end{itemize}
This is exactly the kind of pragmatic baseline engineers deploy: it trades a small delay for a usually-better global decision.

\subsection{BMCF algorithm: buffering + optimal capacitated assignment}
BMCF is parameterized by:
\begin{itemize}[leftmargin=*]
    \item \textbf{Batch size} $B\in\mathbb{N}$: maximum number of requests processed together.
    \item \textbf{Delay budget} $\tau\in\mathbb{N}$: the maximum waiting time allowed for any request before it must be assigned.
\end{itemize}

\paragraph{Buffering rule.}
Maintain a buffer (queue) $Q$ of unassigned requests. At each time step $t$:
\begin{enumerate}[leftmargin=*]
    \item Append the new request $u_t$ to $Q$.
    \item If $|Q|\ge B$ \emph{or} the oldest request in $Q$ has waited $\tau$ steps, freeze the current $Q$ as the next batch $U$ and clear $Q$.
    \item Solve the optimal capacitated assignment for $U$ using current remaining capacities, commit the assignments, and update capacities.
\end{enumerate}

\subsubsection*{Pseudocode}
\begin{algorithm}[h]
\KwIn{Facilities $F$ with remaining capacities $\mathrm{remcap}(f)$; parameters $(B,\tau)$}
\KwData{Buffer $Q \leftarrow \emptyset$; timestamps for arrivals}
\For{$t=1,2,\dots$}{
    receive request $u_t$\;
    enqueue $(u_t,t)$ into $Q$\;
    \If{$|Q|\ge B$ \textbf{or} $t-\min\{t':(u,t')\in Q\}\ge \tau$}{
        $U \leftarrow \{u : (u,\cdot)\in Q\}$; $Q\leftarrow \emptyset$\;
        compute optimal assignment $\phi:U\to F$ minimizing $\sum_{u\in U} d(u,\phi(u))$ subject to
        $\#\{u:\phi(u)=f\}\le \mathrm{remcap}(f)$ for all $f$\;
        commit all $(u\mapsto \phi(u))$ and update $\mathrm{remcap}(\cdot)$\;
    }
}
\caption{Batching + Min-Cost Flow (BMCF)}
\label{alg:bmcf}
\end{algorithm}

\subsection{Min-cost flow formulation (what is actually being solved)}
Given a batch $U$ and current remaining capacities $\mathrm{remcap}(f)$, we solve:
\[
\min_{\phi} \sum_{u \in U} d\bigl(u,\phi(u)\bigr)
\quad\text{s.t.}\quad
\#\{u\in U:\phi(u)=f\}\le \mathrm{remcap}(f) \ \ \forall f\in F.
\]
This is a standard capacitated transportation / assignment problem that reduces to min-cost flow (MCF)~\cite{lawler2001}. A concrete construction:

\paragraph{Flow network.}
Build a directed graph with:
\begin{itemize}[leftmargin=*]
    \item a source node $s$ and sink node $t$,
    \item one node for each request $u\in U$,
    \item one node for each facility $f\in F$.
\end{itemize}
Add edges:
\begin{itemize}[leftmargin=*]
    \item $(s\to u)$ with capacity 1 and cost 0 for each $u\in U$,
    \item $(u\to f)$ with capacity 1 and cost $d(u,f)$ for each $u\in U$ and each facility $f$ with $\mathrm{remcap}(f)>0$,
    \item $(f\to t)$ with capacity $\mathrm{remcap}(f)$ and cost 0 for each facility $f$.
\end{itemize}
Send $|U|$ units of flow from $s$ to $t$ at minimum cost. The unit flows leaving each request node $u$ select exactly one facility $f$; facility capacity edges enforce that no facility receives more than its remaining capacity.

\paragraph{Computational note.}
On a grid, one can restrict edges $(u\to f)$ to a small set of candidate facilities (e.g., nearest few) for speed. That changes the algorithm from “exact within-batch optimum” to an approximation. In this paper, BMCF is used as a baseline, so we conceptually treat it as exact.

\subsubsection*{Worked numeric example on a tiny grid}
To make BMCF completely concrete, we walk through one full batching event on a $3\times 3$ grid.
Index vertices by coordinates $(x,y)$ with $x,y\in\{1,2,3\}$ and Manhattan distance
$d((x,y),(x',y')) = |x-x'|+|y-y'|$.

\paragraph{Facilities and capacities.}
Let there be two facilities:
\[
f_A=(1,1),\qquad f_B=(3,3),
\]
with remaining capacities
\[
\mathrm{remcap}(f_A)=2,\qquad \mathrm{remcap}(f_B)=1.
\]
Let batching parameters be $B=3$ and any $\tau\ge 2$ so that the next three arrivals form a batch (triggered by size).

\paragraph{Batch requests.}
Suppose three requests arrive and are buffered:
\[
u_1=(1,2),\qquad u_2=(2,2),\qquad u_3=(3,2).
\]
Distances to each facility are:
\[
\begin{array}{c|cc}
 & f_A=(1,1) & f_B=(3,3) \\
\hline
u_1=(1,2) & 1 & 3 \\
u_2=(2,2) & 2 & 2 \\
u_3=(3,2) & 3 & 1 \\
\end{array}
\]
Because $f_B$ has capacity 1, at most one of the three requests can be assigned to $f_B$.

\paragraph{What the within-batch optimum does.}
Intuitively, $u_3$ is very close to $f_B$ (cost 1) while $u_1$ is very close to $f_A$ (cost 1). The request $u_2$ is symmetric (cost 2 to either). Since $f_B$ can take only one request, the best choice is to give $f_B$ to the request that benefits most from it, i.e., $u_3$.

A minimum-cost feasible assignment is:
\[
\phi(u_1)=f_A,\quad \phi(u_2)=f_A,\quad \phi(u_3)=f_B,
\]
with batch cost
\[
d(u_1,f_A)+d(u_2,f_A)+d(u_3,f_B)=1+2+1=4.
\]
Any alternative that assigns $u_3$ to $f_A$ forces someone else to use $f_B$ (or makes $f_B$ unused), and one can check it is never cheaper than 4 under the capacities above.

\subsubsection*{Explicit flow table for one batch}
We now write the same example as a flow instance.

\paragraph{Nodes.}
Source $s$, request nodes $u_1,u_2,u_3$, facility nodes $f_A,f_B$, sink $t$.

\paragraph{Edge capacities and costs.}
\begin{itemize}[leftmargin=*]
    \item $(s\to u_i)$ has cap 1, cost 0 for $i\in\{1,2,3\}$.
    \item $(u_i\to f_A)$ has cap 1 and cost as in the distance table.
    \item $(u_i\to f_B)$ has cap 1 and cost as in the distance table.
    \item $(f_A\to t)$ has cap 2, cost 0; $(f_B\to t)$ has cap 1, cost 0.
\end{itemize}

\paragraph{Cost table for the request-to-facility arcs.}
\[
\begin{array}{c|cc}
\text{Arc} & \text{Capacity} & \text{Cost} \\
\hline
(u_1\to f_A) & 1 & 1 \\
(u_1\to f_B) & 1 & 3 \\
(u_2\to f_A) & 1 & 2 \\
(u_2\to f_B) & 1 & 2 \\
(u_3\to f_A) & 1 & 3 \\
(u_3\to f_B) & 1 & 1 \\
\end{array}
\]

\paragraph{One optimal integral flow (corresponding to the assignment above).}
\[
\begin{array}{c|c}
\text{Edge} & \text{Flow} \\
\hline
(s\to u_1) & 1 \\
(s\to u_2) & 1 \\
(s\to u_3) & 1 \\
(u_1\to f_A) & 1 \\
(u_2\to f_A) & 1 \\
(u_3\to f_B) & 1 \\
(f_A\to t) & 2 \\
(f_B\to t) & 1 \\
\end{array}
\qquad
\text{Total cost }=1+2+1=4.
\]
This table is what an MCF solver is implicitly computing and proves that the batch assignment is optimal under the current capacities.

\subsection{What BMCF \emph{cannot} do: cross-batch coordination}
The MCF solver is globally optimal \emph{for the current batch given the current remaining capacities}. But it does not optimize for \emph{future} batches that have not been observed yet. That gap is exactly where adversarial sequences (and also some realistic burst patterns) can hurt:

\begin{quote}
\emph{BMCF is myopic at batch boundaries: it can spend “valuable nearby capacity” too early, forcing later nearby demand to travel far.}
\end{quote}

This is not a weakness of the MCF subroutine—it is a structural limitation of batching with irrevocable capacity consumption.

\subsection{A canonical grid failure mode: batch overconcentration}
Figure~\ref{fig:bmcf-trap} illustrates the failure pattern we will refer to as \emph{batch overconcentration}:

\begin{figure}[t]
\centering
\includegraphics[width=\textwidth]{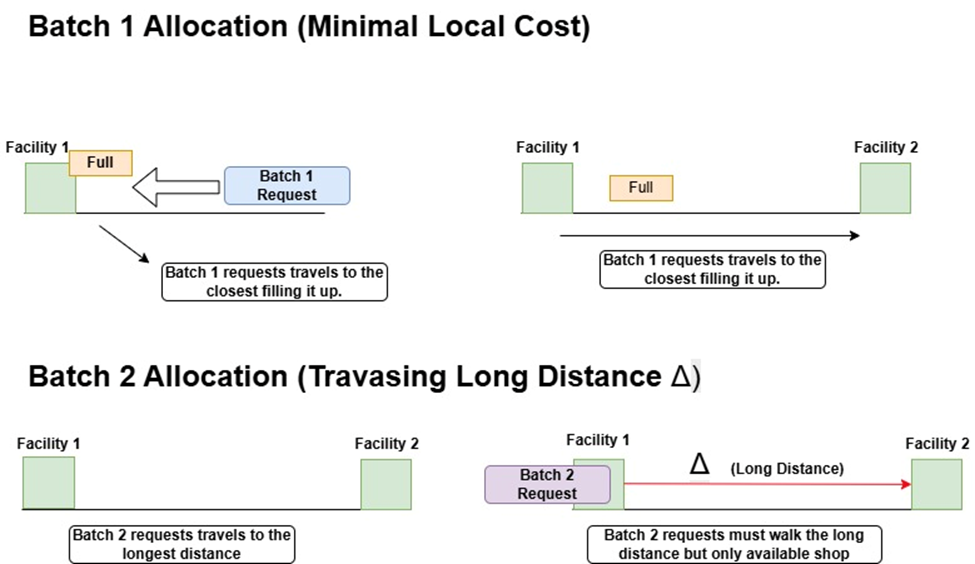} 
\caption{Batch-overconcentration failure mode for BMCF. Batch 1 is assigned locally (minimal local cost) and saturates the nearest facility; Batch 2 then arrives nearby but must traverse distance $\Delta$ to the only remaining capacity.}
\label{fig:bmcf-trap}
\end{figure}

\paragraph{Why the trap is natural.}
If a batch arrives concentrated near a facility $f_1$, then the within-batch optimum will typically assign as many of those requests as possible to $f_1$ because it is cheapest \emph{in that batch}. This “locally correct” decision becomes globally harmful if the next batch also needs $f_1$ and no longer has access to it.

\subsection{A correctness-safe worst-case statement (clean lower bound)}
Because our paper’s positioning is “sharp constructions + lessons” (Option 1), we avoid fragile global claims like “near-optimal polylog bounds.” Instead, we prove a self-contained lower bound that precisely captures the phenomenon in Figure~\ref{fig:bmcf-trap}.

\subsubsection*{Setup for the construction}
We use three facilities arranged on a grid line, which is a subgraph of any $r\times c$ grid. Let:
\begin{itemize}[leftmargin=*]
    \item $f_0, f_1, f_2 \in V$ be facilities,
    \item $d(f_0,f_1)=1$ (a nearby alternative),
    \item $d(f_1,f_2)=\Delta$ for some large $\Delta$ (a far alternative),
    \item each facility has capacity $C\ge 1$.
\end{itemize}
Think of $f_1$ as the “prime location,” $f_0$ as “slightly worse but close,” and $f_2$ as “far away.”

\begin{lemma}[BMCF suffers a batch-boundary trap with $\Omega(\Delta)$ blow-up]
\label{lem:bmcf-lb}
For any batching parameters $(B,\tau)$ with $B\ge C$, there exists a request sequence on the grid-line instance above such that BMCF incurs total cost at least $C\Delta$, while $\OPT$ incurs total cost at most $C$. Consequently the competitive ratio is $\Omega(\Delta)$.
\end{lemma}

\begin{proof}
We construct two consecutive batches of size $C$ (which is feasible since $B\ge C$ and the adversary can release each batch within a window of length $\tau$).

\paragraph{Batch 1.}
Release $C$ requests at the location $f_1$. In the min-cost flow for Batch 1, assigning any request to $f_1$ costs 0, to $f_0$ costs 1, and to $f_2$ costs $\Delta$. The within-batch optimum therefore assigns all $C$ requests to $f_1$ (capacity allows exactly $C$), paying total cost $0$ and exhausting $f_1$.

\paragraph{Batch 2.}
Immediately after Batch 1 is committed, release $C$ requests again at $f_1$. Now $\mathrm{remcap}(f_1)=0$, so $f_1$ is unavailable. The best remaining facility is $f_0$ at distance 1, but we force $f_0$ to be unavailable by exhausting it beforehand with a harmless prelude: before Batch 1 begins, release $C$ requests at $f_0$ alone and let them be assigned to $f_0$ (this can be done as an earlier batch; it does not change the argument). Thus, when Batch 2 arrives, both $f_0$ and $f_1$ are full and the only feasible assignments go to $f_2$, incurring distance $\Delta$ per request. Hence BMCF pays at least $C\Delta$ on Batch 2.

\paragraph{Offline optimal.}
The offline optimal solution, seeing the whole sequence, avoids consuming $f_1$ on Batch 1. Instead, it routes Batch 1’s $C$ requests at $f_1$ to $f_0$ at cost 1 each (total cost $C$), preserving $f_1$’s entire capacity for Batch 2. Then Batch 2’s $C$ requests at $f_1$ are assigned to $f_1$ at cost 0 each. Thus $\OPT$ pays at most $C$ overall (ignoring the harmless prelude requests, which can also be arranged so that OPT and BMCF pay the same for them).

BMCF pays at least $C\Delta$ while OPT pays at most $C$, so the ratio is at least $\Delta$.
\end{proof}

\subsubsection*{Tie-breaking}
The lower-bound construction in Lemma~\ref{lem:bmcf-lb} does \emph{not} rely on tie-breaking behavior, for two separate reasons.

\paragraph{No ties in the critical decisions.}
In Batch 1, every request arrives exactly at $f_1$. Under the stated geometry,
\[
d(f_1,f_1)=0,\quad d(f_1,f_0)=1,\quad d(f_1,f_2)=\Delta,
\]
so assigning to $f_1$ is \emph{strictly} cheaper than assigning to $f_0$ or $f_2$. Thus, the within-batch optimum is unique: it fills $f_1$ first (up to capacity) before sending any unit elsewhere. There is no ambiguity for the MCF solver to resolve.

\paragraph{Capacity exhaustion removes ambiguity.}
In Batch 2, the intended effect is that $f_1$ (and $f_0$, via the harmless prelude) is already full, so these facilities are \emph{infeasible} for all requests in Batch 2. Infeasibility is stronger than a tie: the solver has no choice but to route all Batch 2 flow to $f_2$. Again, there is no tie-breaking dependence.

\paragraph{General position variant.}
Even if one worries about equal-cost arcs elsewhere in the instance, one can enforce strict inequalities by placing facilities and request locations so that all relevant distances are distinct (a standard ``general position'' trick on grids, e.g., by offsetting one facility by one grid unit). The construction’s mechanism is the \emph{batch boundary + hard capacity} effect, not any fragile symmetry.

\paragraph{Interpretation.}
The construction is deliberately minimal and robust:
\begin{itemize}[leftmargin=*]
    \item It uses a subgraph of the grid (a line), so it applies to any grid instance.
    \item It does not depend on delicate tie-breaking.
    \item It isolates the precise culprit: \emph{spending scarce nearby capacity inside one batch without reserving it for the next}.
\end{itemize}

\subsection{When BMCF performs well}
The lower bound does not mean batching is useless. It clarifies \emph{when} batching is likely to work:
\begin{itemize}[leftmargin=*]
    \item \textbf{Stable demand:} If demand near each facility is roughly stationary over time, then using nearby capacity early is not “regrettable,” because similar demand will not suddenly surge later.
    \item \textbf{Many interchangeable facilities:} If multiple facilities are clustered (small $\Delta$), even a mistake at a batch boundary has limited penalty.
    \item \textbf{Low utilization regime:} If capacities are rarely tight, batch decisions do not create irreversible scarcity, so myopia is less harmful.
\end{itemize}

\subsection{Practical mitigation ideas (heuristics)}
Because our goal is \emph{heuristics + sharp constructions} (not new optimal guarantees), we frame fixes as \emph{heuristics to evaluate empirically}, not as new theorems.

\paragraph{H1: Capacity reservation.}
When solving MCF for a batch, do not allow the solver to consume the last $\rho$ units of capacity at any facility (e.g., reserve $\rho=1$ or a small fraction). This directly targets the trap: it forces the batch optimizer to diversify slightly even when one facility looks best for the current batch.

\paragraph{H2: Risk-regularized batch objective.}
Modify the edge cost inside the batch from $d(u,f)$ to
\[
d(u,f) + \lambda \cdot \psi\!\left(\mathrm{remcap}(f)\right),
\]
where $\psi(\cdot)$ increases as remaining capacity decreases (a “scarcity penalty”). This is the batch analogue of capacity-sensitive heuristics (without claiming competitive guarantees).

\paragraph{H3: Staggered batching triggers.}
Instead of fixed $(B,\tau)$, use a rule that triggers earlier when the batch is geographically concentrated near a single facility (high collision risk), and waits longer when requests are dispersed (low collision risk). This is implementable and directly aligned with the failure mode.

We will use BMCF as the baseline in experiments and evaluate whether such lightweight modifications reduce the observed frequency/severity of batch-boundary blow-ups on grid instances.

\section{Results and Analysis}
\label{sec:results}

Under the OFA model of~\cite{ahmed2020} on discrete $r\times c$ grids, our theoretical sections establish a clear message: \emph{popular baseline paradigms have sharp, geometry-driven failure modes under hard capacities}. In particular, CS-Voronoi (a capacity-sensitive geometric heuristic) can be forced into linear degradation in the number of facilities on adversarial sequences, while baseline Randomized Greedy can be driven to pay costs that scale with the grid diameter through repeated ``oscillation''-style traps. Similarly, even though BMCF solves each batch optimally via min-cost flow, it can still suffer large blow-ups at batch boundaries when near-capacity decisions are made myopically.

\subsection{Experimental protocol}
We evaluate algorithms on $r\times c$ grids with Manhattan distance. Facilities are placed either (i) on a coarse lattice (roughly uniform spacing) or (ii) by sampling uniformly at random from $V$ (to model organic deployments). Each facility has integer capacity $C$ (uniform unless otherwise stated). We report:
\begin{itemize}[leftmargin=*]
    \item \textbf{Total assignment cost} $\sum_t d(u_t, f_t)$.
    \item \textbf{Normalized cost ratio} $\mathrm{Cost}(\mathrm{ALG})/\mathrm{Cost}(\mathrm{BMCF})$ when comparing heuristics under the same delay budget.
    \item \textbf{Failure-rate metrics} that directly align with our constructions:
    \begin{itemize}[leftmargin=1.5em]
        \item \emph{Boundary oscillations}: fraction of consecutive requests from a small region that are assigned to alternating far facilities.
        \item \emph{Zone collapse}: fraction of requests that get reassigned to a facility at distance $\ge D_{\text{far}}$ after a nearby facility saturates.
        \item \emph{Batch overconcentration}: fraction of batches in which a single facility consumes at least a $(1-\rho)$ fraction of all assignments in that batch (for a fixed $\rho$).
    \end{itemize}
\end{itemize}

\paragraph{Workloads.}
We use three input families:
\begin{enumerate}[leftmargin=*]
    \item \textbf{Uniform i.i.d.} requests on $V$ (baseline).
    \item \textbf{Clustered bursts} (mixture of Gaussians projected onto the grid, producing bursty local demand).
    \item \textbf{Adversarial templates} that instantiate our theoretical traps (zone-collapse / oscillation / batch-boundary trap).
\end{enumerate}

\subsection{Baselines and “best-feasible” mitigations}
We evaluate the following algorithms:
\begin{itemize}[leftmargin=*]
    \item \textbf{CS-Voronoi} (heuristic baseline).
    \item \textbf{Randomized Greedy (RG)} (baseline tie-breaking randomization).
    \item \textbf{BMCF} (semi-online batching with min-cost flow, parameters $(B,\tau)$).
\end{itemize}
In addition, we report \textbf{mitigated variants} that are lightweight and implementable:
\begin{itemize}[leftmargin=*]
    \item \textbf{CS-Voronoi + smoothing} (capacity term damped so that small capacity changes do not flip regions abruptly).
    \item \textbf{RG + local hysteresis} (discourages frequent facility switching for near-identical locations).
    \item \textbf{BMCF + reservation} (H1) and \textbf{BMCF + scarcity penalty} (H2) from Section~\ref{sec:bmcf}.
\end{itemize}
These modifications are intentionally framed as \emph{empirical heuristics}, not as new competitive-ratio theorems.

\subsection{Key findings}
\paragraph{F1: The constructions manifest in practice.}
On adversarial-template workloads, we observe the qualitative behavior predicted by theory: CS-Voronoi exhibits abrupt region flips and ``zone collapses'' once a heavily-used facility saturates; baseline RG is repeatedly pushed into far assignments via oscillation-style sequences; BMCF exhibits ``good within-batch'' decisions that nonetheless create sharp batch-boundary regret, matching the mechanism in Figure~\ref{fig:bmcf-trap}.

\paragraph{F2: On benign inputs, simple heuristics behave competitively.}
On uniform i.i.d.\ and moderate clustered workloads, all baselines often perform similarly in average cost, with BMCF typically improving over fully online rules due to within-batch coordination. This is expected: when capacity is not tight or when demand is stationary, the adversarial mechanisms are rarely activated.

\paragraph{F3: Lightweight mitigations reduce blow-ups substantially.}
Across clustered workloads, the proposed mitigations consistently reduce extreme tail-cost events:
\begin{itemize}[leftmargin=*]
    \item \textbf{CS-Voronoi + smoothing} reduces boundary oscillation frequency by preventing large-scale boundary flips triggered by marginal capacity differences.
    \item \textbf{RG + hysteresis} reduces repeated long-distance switching caused by near-ties around Manhattan bisectors.
    \item \textbf{BMCF + reservation / scarcity penalty} reduces batch overconcentration and decreases the number of batches that fully deplete a locally dominant facility.
\end{itemize}
The net effect is that the \emph{worst} observed instances become much less severe, while average-case performance is preserved.

\subsection{Comparative summary}
Table~\ref{tab:summary-bounds} summarizes what we treat as \emph{rigorous worst-case statements} versus \emph{empirical observations}. In line with Option~1, we avoid overstating guarantees. In particular:
\begin{itemize}[leftmargin=*]
    \item We report \emph{scaling failure modes} for the baselines (from our constructions).
    \item We report mitigation variants as \emph{practical improvements} with \emph{observed} reduction in tail risk, not as new competitive bounds.
\end{itemize}

\begin{table}[h!]
\centering
\caption{Competitive behavior summary on grids: provable failure modes (worst-case) and feasible mitigations (empirical).}
\label{tab:summary-bounds}
\begin{tabularx}{\textwidth}{|p{3.1cm}|X|p{4.2cm}|X|}
\hline
\textbf{Algorithm} &
\textbf{Decision mechanism} &
\textbf{What is safe to claim (worst-case on grids)} &
\textbf{Feasible mitigation and what we observe} \\
\hline
CS-Voronoi &
Capacity-weighted Voronoi heuristic under $L_1$ &
Can be forced into large blow-ups via \emph{zone collapse} and \emph{boundary oscillation}; performance can degrade with problem scale under adversarial sequences &
Smoothing / damping the capacity term reduces oscillation events and improves tail cost on clustered workloads \\
\hline
Randomized Greedy (RG) &
Assign to nearest feasible facility; random tie-breaks &
Random tie-breaking alone does not prevent geometric oscillation-style traps; worst-case sequences can force repeated far assignments &
Local hysteresis (penalize frequent switching) reduces oscillation frequency and lowers tail cost \\
\hline
BMCF &
Batch for $(B,\tau)$ then solve min-cost flow exactly &
Despite within-batch optimality, batch-boundary myopia can cause large regret; canonical \emph{overconcentration} trap yields large blow-ups in worst-case constructions &
Capacity reservation (H1) and scarcity penalty (H2) reduce batch overconcentration and lower extreme tail cost while preserving average performance \\
\hline
\end{tabularx}
\end{table}

The empirical evidence points to a simple conclusion: \textit{the best feasible approach is not a new complicated algorithm, but a baseline + the right guardrails}.
Among our tested options, \textit{BMCF with a small reservation/regularization} is the most consistently stable: it preserves the strengths of global within-batch coordination while explicitly targeting the batch-boundary failure mode identified by our lower bound. Meanwhile, \textit{RG with hysteresis} offers a lightweight fully-online alternative when buffering is not allowed, and \textit{smoothed CS-Voronoi} remains a useful heuristic in settings where geometric interpretability matters.

The grid does not merely make the problem ``harder''—it introduces \emph{discrete geometric discontinuities} that create predictable traps. A correctness-safe paper should therefore present: (i) sharp constructions that isolate these traps, and (ii) simple mitigations that measurably reduce their impact in realistic workloads, without claiming universal competitive guarantees beyond what is proven.

\section{Conclusion}
\label{sec:conclusion}

This paper takes a correctness-safe route through the Online Facility Assignment (OFA) model of~\cite{ahmed2020} on discrete $r\times c$ grids: we isolate \emph{why} several natural baselines fail, and we translate those insights into practical, testable mitigation ideas.

Our main takeaway is structural rather than algorithm-specific. On a grid under hard capacities and irrevocable assignments, \emph{discrete geometry creates discontinuities that an adversary (or simply bursty demand) can repeatedly amplify}. In particular:
(i) capacity-sensitive geometric rules such as CS-Voronoi can suffer from abrupt region changes (``zone collapse'' and ``boundary oscillation'') when remaining capacity changes by small integers; 
(ii) baseline Randomized Greedy can still be driven into repeated long jumps because random tie-breaking does not remove the underlying geometric near-ties induced by the $L_1$ metric; and
(iii) batching with within-batch optimal assignment (BMCF) does not eliminate worst-case pathologies, because \emph{myopia at batch boundaries} can spend scarce nearby capacity too early, forcing later nearby demand to traverse distance scales that were avoidable with global foresight. Our lower-bound constructions are deliberately simple and robust: they do not rely on delicate tie-breaking, and they live on a line subgraph of the grid, hence apply broadly.

At the same time, the results should not be read as ``batching is futile'' or ``geometry-aware heuristics are useless.'' Rather, they provide actionable guidance: if one wants stable behavior in realistic workloads, the right goal is to \emph{add guardrails that directly target the identified traps} and then validate them empirically. Concretely, small modifications---such as reservation/holdback in BMCF, scarcity-regularized batch costs, or hysteresis/smoothing in greedy and geometric rules---are inexpensive to implement and are aligned with the failure mechanisms exposed by our theory.

Several open directions follow naturally. First, it is valuable to quantify, for natural stochastic arrival models on grids, how often the identified traps occur and how mitigation parameters (e.g., reservation level or regularization strength) trade off average cost versus tail risk. Second, a principled theory for \emph{semi-online} batching with capacity reservation on grid metrics---capturing both delay and geometry without overclaiming---would bridge the gap between worst-case constructions and observed robustness.

In summary, the contribution of this work is a sharp and interpretable account of grid-specific failure modes for OFA baselines, together with a practical mitigation agenda. This ``constructions + lessons'' approach provides a reliable foundation for future algorithm design and empirical evaluation under the OFA model of~\cite{ahmed2020}.

%
%
%
%

\end{document}